\documentclass{amsart}

\usepackage{amsaddr}
\usepackage{amsthm, amsmath, dsfont}
\usepackage[utf8]{inputenc}
\usepackage{tikz}

\newtheorem{thm}{Theorem}
\newtheorem{lem}{Lemma}

\newcommand{\nop}[1]{}

\DeclareMathOperator{\Out}{Out}
\DeclareMathOperator{\In}{In}

\begin{document}
\title[On the Radius of Nonsplit Graphs and Information Dissemination]{On the Radius of Nonsplit Graphs and\\Information Dissemination in Dynamic Networks}

\author{Matthias F\"ugger}
\address{CNRS, LSV, ENS Paris-Saclay, Université Paris-Saclay, Inria}
\email{mfuegger@lsv.fr}
\author{Thomas Nowak}
\address{Université Paris-Sud}
\email{thomas.nowak@lri.fr}
\author{Kyrill Winkler}
\address{ECS, TU Wien}
\email{kwinkler@ecs.tuwien.ac.at}
\thanks{
This work has been supported by the Austrian Science Fund (FWF) projects
191020-2000 (RiSe), 191020-2002 (ADynNet) and SIC (P26436),
the CNRS project PEPS DEMO, and the Institut Farman project Dicimus.}

\begin{abstract}
A nonsplit graph is a directed graph where each pair of nodes has a common incoming neighbor.
We show that the radius of such graphs is in $O(\log \log n)$, where $n$ is the number of
  nodes.
We then generalize the result to products of nonsplit graphs.

The analysis of nonsplit graph products has direct implications in the context of distributed
  systems, where processes operate in rounds and communicate via message passing in each round:
  communication graphs in several distributed systems naturally relate to nonsplit graphs
  and the graph product concisely represents relaying messages in such networks.
Applying our results, we obtain improved bounds on the dynamic radius of such networks, i.e.,
  the maximum number of rounds until all processes have received a message from a common
  process, if all processes relay messages in each round.
We finally connect the dynamic radius to lower bounds for achieving consensus in
  dynamic networks.
\end{abstract}


%
\maketitle

\section{Introduction}

Consider a distributed system of $n \ge 1$ processes that operate in lock-step synchronous rounds.
Let $[n] = \{1, \dots, n\}$ be the set of processes.
In a round, each process broadcasts a message and receives messages from a subset
of other processes, specified by the directed communication graph $G = ([n], E)$
whose nodes are the processes and there is an edge $(u,v)$ in $E$ if and only if
process $v$ receives the message sent by process $u$.

The radius of communication graph $G$ is the number of rounds until all processes
have (transitively) received a message from a common process.
Its value thus poses a lower bound on the number of rounds until information, originating at a
single process, can be spread over the entire network.
Related applications are from disease spreading and opinion dynamics, where the radius
is a lower bound on the rounds it takes to spread a disease or
an opinion that originates at a single agent.

Of particular interest in distributed computing are networks that potentially change
during execution of an algorithm, be it due to faulty processes, faulty links,
mobility of the involved agents, etc.; see, e.g., \cite{Lyn96} for 
a comprehensive overview.
We thus generalize the investigation of the radius of a communication graph $G$
to the dynamic radius of a sequence of communication graphs $G_1, G_2, \dots$.
Here, it is assumed that in the above scenario of broadcasting distributed
processes, the communication graph for round $r \ge 1$ is $G_r$.
The {\em dynamic radius of the sequence $G_1, G_2, \dots$\/} is the
number of rounds until all processes have (transitively) received a message from a
common process.

\subsection{Radius of Nonsplit Digraphs}
A nonsplit digraph is a 
  directed graph where each pair of nodes has {\em at least one\/}
  common incoming neighbor.
In this work, we study the radius of nonsplit digraphs:
  with $\ell(u,v)$ denoting the length of the shortest path from node $u$ to node $v$,
  the radius is $\min_u \max_{v} \ell(u,v)$.

In the undirected case, the radius is trivially bounded by the diameter of the graph,
  which is $2$ in the case of nonsplit graphs.
Undirected graphs where each pair of nodes has {\em exactly one\/} common neighbor,
  have been studied by Erd\H{o}s et al.~\cite{ERS66}, who showed that they are exactly
  the windmill graphs, consisting of triangles that share a common node.
Thus, their radius is $1$.

As demonstrated by the example in Figure~\ref{fig:graph} with radius $3$, these bounds do not
  hold for nonsplit digraphs.
We will show the following upper bound:
\begin{thm}\label{thm:radius}
The radius of a nonsplit digraph with $n$ nodes is in $O(\log \log n)$.
\end{thm}

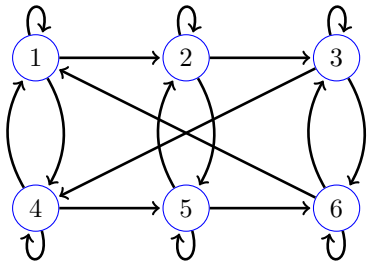
\begin{figure}[hbt]
\centering
\begin{tikzpicture}[shorten >=1pt, auto, node distance=3cm]
    \tikzstyle{node_style} = [circle,draw=blue]
    \tikzstyle{edge_style} = [->,draw=black, line width=1]
    
    \node[node_style] (v1) at (0,0) {1};
    \node[node_style] (v2) at (2,0) {2};
    \node[node_style] (v3) at (4,0) {3};

    \node[node_style] (v4) at (0,-2) {4};
    \node[node_style] (v5) at (2,-2) {5};
    \node[node_style] (v6) at (4,-2) {6};
    
    \draw[edge_style]  (v1) to (v2);
    \draw[edge_style]  (v2) to (v3);

    \draw[edge_style]  (v4) to (v5);
    \draw[edge_style]  (v5) to (v6);

    \draw[edge_style]  (v6) to (v1);
    \draw[edge_style]  (v3) to (v4);

    \draw[edge_style]  (v1) to[bend left] (v4);
    \draw[edge_style]  (v4) to[bend left] (v1);
	\draw[edge_style]  (v2) to[bend left] (v5);
    \draw[edge_style]  (v5) to[bend left] (v2);
	\draw[edge_style]  (v3) to[bend left] (v6);
    \draw[edge_style]  (v6) to[bend left] (v3);

    \draw[edge_style]  (v1) to[loop above] (v1);
    \draw[edge_style]  (v2) to[loop above] (v2);
    \draw[edge_style]  (v3) to[loop above] (v3);
    \draw[edge_style]  (v4) to[loop below] (v4);
    \draw[edge_style]  (v5) to[loop below] (v5);
    \draw[edge_style]  (v6) to[loop below] (v6);

    \end{tikzpicture}
\caption{Nonsplit digraph with radius $3$. For example, node $1$ and $6$ have common incoming neighbor $6$,
  while nodes $1$ and $5$ have node $4$ as common incoming neighbor.}
\label{fig:graph}
\end{figure}

\subsection{Communication over Nonsplit Digraphs}
Nonsplit digraphs naturally occur as communication graphs in classical fault-models and
  as models for dynamic networks.

In fact, several classical fault-models were shown to lead to nonsplit
  communication graphs~\cite{CFN15:ICALP}, among them link failures, as considered in~\cite{SW89},
  and asynchronous message passing systems with crash failures~\cite{Lyn96}.
Nonsplit digraphs thus represent a convenient abstraction to these classical fault-models.
We will show in Section~\ref{sec:async} that nonsplit digraphs arising from  
  the classical model of asynchronous messages and crashes have dynamic radius at
  most~$2$.

The study of nonsplit digraphs is also motivated by the study of a central problem in distributed computing:
  Agreeing on a common value in a distributed system is a problem that lies at the heart of many
  distributed computing problems, occurs in several flavors, and thus received considerable
  attention in distributed computing.
However, even modest network dynamics already prohibit solvability of exact consensus,
  where agents have to decide on a single output value that is within the range of the agents'
  initial values~\cite{SW89}.
For several problems, e.g., distributed control, clock synchronization, load balancing, etc.,
  it is sufficient to asymptotically converge to the same value (asymptotic consensus),
  or decide on values not too far from each other (approximate consensus).
Charron-Bost et al.\ \cite{CFN15:ICALP} showed that both problems are
  solvable efficiently in the case of communication graphs that may vary arbitrarily, but are required to
  be nonsplit.
They further showed that in the more general case where all communication graphs are required to
  contain a rooted spanning tree, one can in fact simulate nonsplit communication graphs,
  leading to efficient algorithms for asymptotic and approximate consensus.

Motivated by this work on varying communication graphs, we will show that the following generalization
  of Theorem~\ref{thm:radius} holds:
\begin{thm}\label{thm:radius2}
The dynamic radius of a network on~$n$ nodes
whose communication graphs are all nonsplit is $O(\log\log n)$.
\end{thm}

Traditionally, information dissemination is studied w.r.t.\ either all-to-all message
relay or the time it takes for a fixed process to broadcast its
message to everyone \cite{HLPR13,HHL88}.
In dynamic networks with nonsplit communication graphs, however, such
strong forms of information dissemination are impossible.
This can easily be seen by constructing appropriate sequences of star graphs
(with self-loops), which are a degenerate case of a nonsplit graph.

One-to-all broadcast of \emph{some} process, on the other hand, is readily achieved in such networks, which is why we focus on this characteristic here.
While it is certainly not as universal as the previously mentioned primitives,
it turns out that this type of information dissemination is crucial to the
termination time of certain consensus algorithms based on vertex-stable root
components \cite{WSS16}.
Furthermore, we show the following theorem, relating the dynamic radius and the
termination time of arbitrary consensus algorithms:

\begin{thm}
  If the dynamic radius of a sequence of communication graphs is $k$, then, in
  every deterministic consensus algorithm, some process has not terminated before time $k$.
\end{thm}

Finally, we note that the dynamic radius is also an upper bound for a single process aggregating
the data of all other processes, when we use the
alternative interpretation of an edge $(i, j)$ in a communication graph
as a message sent by $j$ and received by $i$. 
Even though this might not be the desired form of data aggregation in a
standard setting, in a scenario where the communication is so constrained that
aggregation by an \emph{a priori} selected process is simply unobtainable, such a weak
form
might still be useful to transmit the collected data to a dedicated sink at
regular intervals, for example.

We give a brief overview on related work in the next section.

\subsection{Related Work}
\label{sec:relwork}

Information dissemination among an ensemble of $n$ participants is a fundamental
question that has been studied in a grand variety of settings and flavors (see
\cite{FL94,HLPR13,HHL88,HKMP96} for various reviews on the topic).
While traditional approaches usually assume a static underlying network topology, with
the rise of pervasive wireless devices,
more recently, focus has shifted to dynamically changing network
topologies \cite{KLO10:STOC}.
A useful way of viewing the distribution of information is to denote the pieces
of information that should be shared among the participants as tokens.
For instance, the all-to-all token dissemination problem investigates the complete
dissemination of $n$ initially distributed tokens.
This problem was studied in \cite{KLO10:STOC} with a focus on bounds for
the time complexity of the problem, i.e., how long it takes at least, resp.\ at
most, until $n$ tokens have been received by everyone.
Here, the participants employed a token-forwarding algorithm mechanism, where
tokens are stored and forwarded but not altered.

In the model of \cite{KLO10:STOC}, it was assumed that the communication graphs
are connected and undirected.
For this, a lower bound of $\Omega(n \log n)$ and an upper bound of $O(n^2)$ for
all-to-all token dissemination was established in the case where $n$ is unknown
to the participants, they have to terminate when the broadcast is finished, and
the system is $1$-interval connected, i.e., the communication graphs are
completely independent of each other.
In contrast, if the communication graphs are connected, directed, and rooted,
in the worst case only one of the tokens may ever be delivered to all
participants.
This can be seen, for example, when considering a dynamic graph
that produces the same line graph for every round.
We note that this example also provides a trivial lower bound of
$\Omega(n)$ rounds until one token is received by everyone for the first
time.
As far as we are aware, the best such lower bound was established in
\cite[Theorem 4.3]{SZS17} to be $\lceil (3n-1)/2-2 \rceil$ rounds.
Studying directed graphs is desirable as they represent a weaker, more general
model and wireless communication is often inherently directed, for example due
to localized fading or interference phenomena \cite{SBB12,GKFS10:perf} such as
the capture effect or near-far problems \cite{WJC00}.

In \cite{CBS09}, it was shown that the dynamic radius of a sequence of
arbitrary nonsplit communication graphs is $O(\log n)$.
Later, it was shown in \cite{CFN15:ICALP} that the product
of any $n-1$ rooted communication graphs is nonsplit.
Put together, this means that the dynamic radius of a sequence of arbitrary
rooted graphs is $O(n \log n)$.
More recently, \cite{SZS17} provided an alternative proof for this fact that
does not rely on the reduction to nonsplit graphs but instead uses a notion
of influence sets similar to \cite[Lemma 3.2.(b)]{KLO10:STOC}.
In addition to this, \cite{SZS17} provided linear $O(n)$ bounds in sequences of
rooted trees with a constant number of leaves or inner nodes, established a
dependency on the size of certain subtrees in sequences of rooted trees where
the root remains the same, and investigated sequences of undirected trees.

\section{Model and Definitions} \label{sec:model}

We start with some definitions motivated by the study of information dissemination
  within a distributed system of $n$ processes that operate in discrete,
  lock-step synchronous communication rounds.
Starting with information being available only locally to each process,
  processes broadcast and receive information tokens in every round.
We are interested in the earliest round where all processes have received an information token
  from a common process.

Clearly, the dissemination dynamics depends on the dynamics of the underlying network.
For this purpose we define:
A {\em communication graph on $n$ nodes\/} is a directed graph $G=(V,E)$ with self-loops and
the set of nodes $V = [n] = \{1,2,\dots,n\}$.
For $i \in [n]$, let $\In_i(G) = \{ j \in [n] \mid (j, i) \in E \} $ denote the
set of in-neighbors of $i$ in $G$ and $\Out_i(G) = \{ j \in [n] \mid (i, j) \in
E \}$ denote its set of out-neighbors.
Intuitively, communication graphs encode successful message reception within a round:
  an edge from $i$ to $j$ states that $j$ received the message broadcast by $i$ in this
  round.

A node $i\in [n]$ is called a {\em broadcaster in~$G$\/} if it has an edge
to all nodes, i.e.,
$\forall j\in [n]\colon (i,j)\in E$.

A communication graph $G=([n],E)$ is {\em nonsplit\/} if
every pair of nodes has a common incoming neighbor, i.e.,
\[
\forall i,j\in[n]\ \exists k\in[n] \colon (k,i)\in E \wedge (k,j)\in E
\enspace.
\]

Given two communication graphs $G = ([n],E_G)$ and $H = ([n],E_H)$ on $n$
nodes, define their
{\em product graph\/} as $G\circ H = ([n], E_{G\circ H})$ where
\[
(i,j) \in E_{G\circ H}
\iff
\exists k\in [n] \colon (i,k)\in E_G \wedge (k,j)\in E_H
\enspace.
\]
The empty product is equal to the communication graph $([n],E_\perp)$ which
contains the self-loops $(i,i)$ for all nodes~$i$ and no other edges.
The graph product we use here is motivated by information dissemination within
  distributed systems of processes that continuously relay information tokens that they received:
  if $k$ received $i$'s information token in a round, and $j$ received $k$'s information token
  in the next round, then $j$ received $i$'s information token in the macro-round formed
  by these two successive rounds.

Motivated by modeling communication networks that potentially
  change in each round, we call each infinite sequence
$\mathcal{G} = (G_1,G_2,G_3,\dots)$ of communication
graphs on~$n$ nodes a {\em communication pattern on~$n$ nodes}.
For every node $i \in [n]$, define the {\em broadcast time $T_i(\mathcal{G})$
of node~$i$  in~$\mathcal{G}$\/} as the minimum~$t$ such that~$i$ is a
broadcaster in the product of the first~$t$ communication graphs of~$\mathcal{G}$.
If no such~$t$ exists, then $T_i(\mathcal{G}) = \infty$.
The {\em dynamic radius $T(\mathcal{G})$ of~$\mathcal{G}$\/} is the minimal
broadcast time of its nodes, i.e.,
$\displaystyle T(\mathcal{G}) = \min_{i\in[n]} T_i(\mathcal{G})$.
Note that $T(\mathcal{G})$ is the earliest time, in terms of rounds,
  until that all nodes have received an information token from a common node,
  given that the communication pattern is $\mathcal{G}$.

A {\em network on $n$ nodes\/} is a nonempty set of communication patterns
  on~$n$ nodes; modeling potential uncertainty in a dynamic communication network.
A network's {\em dynamic radius\/} is defined as the supremum over all dynamic
  radii of its communication patterns, capturing the worst-case of information
  dissemination within this network.

\section{The Dynamic Radius of Nonsplit Networks}

In this section we show an upper bound on the dynamic radius of nonsplit
networks.

During the section,
let $\mathcal{G} = (G_1,G_2,G_3,\dots)$ be a communication pattern on~$n$ nodes
in which every communication graph~$G_t$ is nonsplit.

In order to prove an upper bound on the dynamic radius of $\mathcal{G}$,
  we will prove the existence of a relatively small set of nodes that "infects" all other nodes
  within only $O(\log\log n)$ rounds.
Iteratively going back in time, it remains to show that any such set is itself "infected"
  by successively smaller and smaller sets within $O(\log\log n)$ rounds,
  until we reach a single node.
It follows that this single node has "infected" all nodes with its information token
  after $O(\log\log n)$ rounds.

Note that the strategy to follow "infection" back in time rather than
  consider the evolution of infected sets over time is essential in our proofs:
  it may very well be that a certain set of infected nodes cannot infect
  other nodes from some time on, since it only has incoming edges from 
  nodes not in the set in all successive communication graphs. 
Going back in time prevents us to run into such dead-ends of infection.

For that purpose we define: Let $U,W\subseteq [n]$ be sets of nodes.
We say that~$U$ {\em covers~$W$ in communication graph $G=([n],E)$\/} if
for every $j\in W$ there is some $i\in U$ that has an edge to~$j$, i.e.,
$\forall j\in W\ \exists i\in U\colon (i,j)\in E$.

Now let $0 < t_1 \leq t_2$.
We say that {\em $U$ at time~$t_1$ covers~$W$ at time~$t_2$\/} if~$U$ covers~$W$
in the product graph $G_{t_1} \circ G_{t_1+1}\circ \cdots\circ G_{t_2-1}$.

Note that~$U$ at time~$t$ covers~$U$ at time~$t$ for all sets $U\subseteq [n]$
and all $t\geq 1$, by definition of the empty product as the digraph with only
self-loops.

We first show that the notion of covering is transitive:

\begin{lem}\label{lem:cover:combination}
Let $0 < t_1 \leq t_2 \leq t_3$ and let $U,W,X\subseteq [n]$.
If~$U$ at time~$t_1$ covers~$W$ at time~$t_2$, and~$W$ at time~$t_2$ covers~$X$
at time~$t_3$, then $U$ at time~$t_1$ covers~$X$ at time~$t_3$.
\end{lem}
\begin{proof}
By definition, for all $k\in W$ there is some $i\in U$
such that $(i,k)$ is an edge of the product graph
$G_{t_1}\circ\cdots\circ G_{t_2-1}$.
Also, for all $j\in X$ there is some $k\in W$
such that $(k,j)$ is an edge of the product graph
$G_{t_2}\circ\cdots\circ G_{t_3-1}$.

But, by the associativity of the graph product, this means that for all $j\in X$
there exists some $i\in U$ such that $(i,j)$ is an edge in the product graph
\[
\big(G_{t_1}\circ\cdots\circ G_{t_2-1}\big)
\circ
\big(G_{t_2}\circ\cdots\circ G_{t_3-1}\big)
=
G_{t_1}\circ\cdots\circ G_{t_3-1}
\enspace.
\]
That is, $U$ at time~$t_1$ covers~$X$ at time~$t_3$.
\end{proof}

\begin{lem}\label{lem:log2:ceil}
For all $x \ge 1$ we have 
$\big\lceil\log_2 x\big\rceil = \big\lceil\log_2 \lceil x\rceil \big\rceil$.
\end{lem}
\begin{proof}
We have
$\lceil \log_2 x\rceil = \min\{k\in \mathds{Z} \mid x \leq 2^k \}$ if $x \ge 1$.
Now, noting that the inequality $x\leq p$ is equivalent to 
$\lceil x\rceil \leq p$ whenever~$p$ is an integer
concludes the proof.
\end{proof}

\begin{lem}\label{lem:log2:twosets}
Let~$m$ and~$n$ be positive integers such that $\lvert m-n\rvert\leq1$.
Then $\lceil\log_2(m+n)\rceil \geq \lceil \log_2 m \rceil + 1$.
\end{lem}
\begin{proof}
W.l.o.g., let $n\geq m-1$.
We distinguish between two cases for positive integer $m$:

\medskip
\noindent (i) If $m = 1$ then $n \in \{1,2\}$, and we immediately obtain
the lemma from $\lceil\log_2(2)\rceil \geq \lceil \log_2 1 \rceil + 1$
and $\lceil\log_2(3)\rceil \geq \lceil \log_2 2 \rceil + 1$.

\medskip

\noindent (ii) Otherwise, $m \ge 2$. From $n\geq m-1$
we deduce $m+n\geq 2m-1$.
This implies
\[
\big\lceil \log_2 (m+n) \big\rceil 
\geq
\big\lceil \log_2 (2m-1) \big\rceil 
=
\left\lceil \log_2 \left(m-\frac{1}{2}\right) \right\rceil + 1
\enspace.
\]
We are hence done if we can show
$\big\lceil \log_2(m-\frac{1}{2})\big\rceil = \lceil \log_2 m \rceil$.
But this is just Lemma~\ref{lem:log2:ceil} with $x=m-\frac{1}{2} \ge 1$.
\end{proof}

\begin{lem}\label{lem:log2:partition}
Let~$n$ and~$m$ be positive integers such that $n \ge m$.
Then there exist positive integers $n_1,n_2,\dots,n_m$ such that
$n = n_1 + \cdots + n_m$ and
$\lceil \log_2 \frac{n}{m} \rceil \geq \lceil\log_2 n_i\rceil$
for all $1\leq i\leq m$.
\end{lem}
\begin{proof}
Let $n = km + r$ with $k,r\in \mathds{Z}$ and $0\leq r < m$ be the integer
division of~$n$ by~$m$.
Set $n_1=n_2=\cdots=n_{r}=k+1$ and $n_{r+1}=n_{r+2}=\dots=n_m = k$.

By Lemma~\ref{lem:log2:ceil}, we have 
\[
\lceil \log_2 n_i \rceil
\leq
\left\lceil \log_2 \left(k + \left\lceil \frac{r}{m} \right\rceil\right) \right\rceil
=
\left\lceil \log_2 \left(k + \frac{r}{m}\right) \right\rceil
=
\left\lceil \log_2 \frac{n}{m} \right\rceil
\]
for all $1\leq i\leq m$.
\end{proof}

We continue with the following generalization of a result by Charron-Bost and
  Schiper~\cite{CBS09}.
In particular ($m = 1$), it shows that any set of nodes can be "infected" by a
  single node, such that the set of infected nodes grows exponentially in size
  per round.

\begin{lem}\label{lem:log:bound}
Let $W\subseteq[n]$ be nonempty and $m$ be a positive integer.
If $t_2-t_1 \geq \log_2 \frac{\lvert W\rvert}{m}$, then 
there exists some $U\subseteq [n]$ with $\lvert U \rvert \leq m$
such that~$U$ at time~$t_1$ covers~$W$ at time~$t_2$.
\end{lem}
\begin{proof}
Using Lemma~\ref{lem:log2:partition}, we can assume without loss of generality
that $m=1$.
We proceed by induction on $t_2-t_1 \ge 0$.

\medskip

\noindent Base case:
If $t_2-t_1=0$, i.e., $t_1=t_2$, then $\lvert W\rvert = 1$ and the statement is
trivially true since we can choose $U=W$.

\medskip

\noindent Inductive step: Now let $t_2 - t_1 \geq 1$.
Let $W = W_1\cup W_2$ such that 
$\big\lvert \lvert W_1\rvert - \lvert W_2\rvert \big\rvert \leq 1$.
Using Lemma~\ref{lem:log2:twosets}, we see that 
$t_2-(t_1+1) \geq \lceil \log_2\lvert W_s\rvert\rceil$ for $s\in\{1,2\}$.
By the induction hypothesis, there hence exist nodes~$j_1$ and~$j_2$
that at time $t_1+1$ cover~$W_1$ and~$W_2$, respectively.
But now, using the nonsplit property of communication graph $G_{t_1}$, 
we see that there exists a node~$i$ that covers $\{j_1,j_2\}$ in~$G_{t_1}$.
An application of Lemma~\ref{lem:cover:combination} concludes the proof.
\end{proof}

Note that Lemma~\ref{lem:log:bound}, by choosing $W = [n]$ and $m = 1$,
  provides an upper bound on the dynamic radius of $O(\log n)$.
To show an upper bound of $O(\log\log n)$, we will apply this lemma
  only for the early infection phase of $O(\log\log n)$ rounds, and
  use a different technique, by the next two lemmas,
  for the late phase.

\begin{lem}\label{lem:inverse:image:size}
Let~$U$ and~$W$ be finite sets with $\lvert U\rvert =k$, $\vert W\rvert = n$,
and $f:\binom{U}{\lfloor \log n\rfloor} \to W$.
If $n\geq 8$, 
then there exists some $w\in W$ such that
$\left\lvert \bigcup f^{-1}[\{w\}] \right\rvert \geq k/e^4$.
\end{lem}
\begin{proof}
By the pigeonhole principle, we get the existence of some $w\in W$ with
\begin{equation}\label{eq:lem:inverse:image:size:lower}
\begin{split}
\left\lvert f^{-1}[\{w\}] \right\rvert
\geq
\frac{\binom{k}{\lfloor \log n\rfloor}}{n}
\geq
\frac{k^{\lfloor \log n\rfloor}}{n \lfloor \log n\rfloor^{\lfloor \log n\rfloor}}
\enspace.
\end{split}
\end{equation}
Write $M = \bigcup f^{-1}[\{w\}]$ and $m = \lvert M \rvert$.
Since $S \in \binom{M}{\lfloor \log n \rfloor}$ for all $S\in f^{-1}[\{w\}]$, we
have
\begin{equation}\label{eq:lem:inverse:image:size:upper}
\begin{split}
\lvert f^{-1}[\{w\}]\rvert 
\leq 
\binom{m}{\lfloor \log n \rfloor}
\leq 
\frac{m^{\lfloor \log n \rfloor} e^{\lfloor \log n \rfloor}}{{\lfloor \log n
\rfloor}^{\lfloor \log n \rfloor}}
\leq
\frac{m^{\lfloor \log n \rfloor} n}{{\lfloor \log n
\rfloor}^{\lfloor \log n \rfloor}}
\enspace.
\end{split}
\end{equation}
Combining~\eqref{eq:lem:inverse:image:size:lower}
and~\eqref{eq:lem:inverse:image:size:upper},
we get
\begin{equation}\label{eq:lem:inverse:image:size:conc}
\begin{split}
m 
& \geq 
\frac{k}{n^{2/\lfloor \log n \rfloor}}
=
\frac{k}{e^{2\log n/\lfloor \log n \rfloor}}
\geq
\frac{k}{e^{2\log n/(\log n - 1)}}
\\
& =
\frac{k}{e^{2/(1 - 1/\log n)}}
\geq
\frac{k}{e^{2/(1 - 1/2)}}
=
\frac{k}{e^4}
\end{split}
\end{equation}
where we used $\log n \geq \log 8 \geq 2$.
This concludes the proof.
\end{proof}

The next lemma shows that during late infection, nodes are infected faster than
   than exponential as provided by Lemma~\ref{lem:log:bound} for the early phase:

\begin{lem}\label{lem:log:at:depth:loglog}
There exists some~$C>0$ such that
for all $t\geq 1$ there exists a set of at most $C\log n$ nodes
that at time~$t$ covers
the set~$[n]$ of all nodes at time $t + \lceil \log_2 \log n\rceil$.
\end{lem}
\begin{proof}
For every set $A\in\binom{V}{\lfloor \log n\rfloor}$ of $\lfloor \log n \rfloor$
nodes, let $f(A) \in V$ be a node that at time~$t$ covers~$A$ at time~$t+\lceil \log_2 \log
n \rceil$, which exists by Lemma~\ref{lem:log:bound}.

We recursively define the following sequence of nodes~$v_i$, $\geq1$ and sets
of nodes~$V_i$, $i\geq 0$:
\begin{itemize}
\item $V_0 = V$
\item For $i\geq 1$, we choose~$v_i$ such that $\lvert \bigcup
f^{-1}[\{v_i\}]\rvert \geq \lvert V_{i-1}\rvert/e^4$, which exists by
Lemma~\ref{lem:inverse:image:size}, and $V_i = V_{i-1} \setminus
\bigcup f^{-1}[\{v_i\}]$.
\end{itemize}

Note that, setting $r = 1 + \log n / \log\frac{e^4}{e^4-1}$, we have $V_r =
\emptyset$.
Hence the set $\{v_1,\dots,v_r\}$ at time~$t$ covers all nodes at time 
$t + \lceil \log_2 \log n \rceil$.
Noting $r = O(\log n)$ concludes the proof.
\end{proof}

We are now ready to combine Lemma~\ref{lem:log:bound} for the early phase
  and Lemma~\ref{lem:log:at:depth:loglog} for the late phase to prove the main
  result of this section.

\begin{thm}\label{thm:nonsplit:loglog}
The dynamic radius of a network on~$n$ nodes
whose communication graphs are all nonsplit is $O(\log\log n)$.
\end{thm}
\begin{proof}
Let $t = \lceil\log_2  (C \log n)\rceil$ where~$C$ is the constant from 
Lemma~\ref{lem:log:at:depth:loglog}.

By Lemma~\ref{lem:log:at:depth:loglog}, there is a set~$A$ of nodes with
$\lvert A\rvert\leq C\log n$ that at time~$t$ covers all nodes at
time $t +\lceil \log_2 \log n \rceil$.
By Lemma~\ref{lem:log:bound}, a single node at time~$1$ covers~$A$ at
time~$t$.

Combining both results via Lemma~\ref{lem:cover:combination} shows that a
single node at time~$1$ covers all nodes at
time 
$\lceil\log_2  (C \log n)\rceil + \lceil \log_2 \log n \rceil = 
O(\log\log n)$.
\end{proof}

\section{Nonsplit Networks from Asynchronous Rounds}
\label{sec:async}

We now show that in an important special case of nonsplit networks, namely
  those evolving from distributed algorithms that establish a round structure over
  asynchronous message passing in the presence of crashes, the dynamic radius is
  at most $2$.

In the classic asynchronous message passing model with crashes, it is assumed
that all messages sent have an unbounded but finite delay until they are
delivered.
Furthermore, processes do not operate in lock-step but may perform their
computations at arbitrary times relative to each other.
In addition, some processes may be faulty in the sense that they are prone to
crashes, i.e., they may seize to perform computations at an arbitrary point in
time.

This means that in a system where up to $f$ processes may be faulty, in order
to progress in a distributed algorithm, a process may wait until it received a
message from $n-f$ different processes but no more: If a process waits for
a message from $> n-f$ different processes,
but there were in fact $f$ crashes, this process will wait forever.
For this reason, algorithms for this asynchronous model often employ the concept
of an \emph{asynchronous round}, sometimes realized as a local round counter variable
$r_i$, which is held by each process $i \in [n]$ and appended to every message.
A process $i$ increments $r_i$ only if it received a message
containing a round counter $\ge r_i$ from $n-f$ different processes.

One may now ask how fast information tokens can spread in such a distributed system,
  again, if processes repeatedly receive and forward information tokens.

For that purpose we consider the network whose communication patterns are induced by $n$ processes
  communicating in asynchronous rounds.
When deriving the communication graph of such an asynchronous round, we get a digraph where
  each process has at least $n-f$ incoming neighbors.
In this sense, the round $t$ communication graph $G_t$ represents, for each process $i$, a set of $n-f$ processes that managed to send a message to $i$, containing a round counter $\ge t$, before they crashed.
It is important to note that an edge $(i, j)$ in $G_t$ represents that $j$ received a message from $i$ that contained a round
counter $r_i \ge t$, but not necessarily $r_i=t$.

Below, we study the case where $n > 2f$, i.e., a majority of the processes is
correct.
This implies that the sets of incoming neighbors of any two processes in a
communication graph have a non-empty intersection, which means that the
communication graph is nonsplit.
In fact, if $n \leq 2f$, then the network can be disconnected into two disjoint sets of processes that do not
receive messages from the other until termination of the algorithm.
Below, we establish a constant upper bound on the dynamic radius of this
important class of nonsplit graphs.

\begin{thm}
Let $f \ge 0$, $n > 2f$, and $(G_r)_{r \ge 1}$ be a sequence of communication graphs
  with $\In_i(G_r) \ge n-f$ for all $r$ and all $i$.
The dynamic radius of $(G_r)_{r \ge 1}$ is at most $2$.
\end{thm}
\begin{proof}
To show the bound on the radius we prove that there exists a node $m$, that will be
  the center that realizes the dynamic radius, i.e.,
\begin{align}
  \exists m \in [n]\ \forall i \in [n]\ \exists j \in [n]: \, m \in \In_j(G_1) \wedge j \in \In_i(G_2)\,.\label{eq:s1}
\end{align}

Equation \eqref{eq:s1} now follows from
\begin{align}
  \exists m \in [n] :\, \lvert \Out_m(G_1) \rvert \ge f+1\,,\label{eq:s2}
\end{align}
by the following arguments: Equation~\eqref{eq:s2} states that the information at
  $m$ has been transmitted to at least $f+1$ nodes.
By assumption $\In_i(G_2) \ge n-f$ for all $i \in [n]$.
Thus each $i$ must have an incoming
  neighbor $j$ in digraph $G_2$ such that $j \in \Out_m(G_1)$;
  equation~\eqref{eq:s1} follows.

It remains to show \eqref{eq:s2}.
Suppose that the equation does not hold, i.e.,
\begin{align}  
  \forall j \in [n] : \, |\Out_j(G_1)| \le f\,.\label{eq:tocontradict}
\end{align}
By assumption on digraph $G_1$, we have
\begin{align}
\sum_{i \in [n]}|\In_i(G_1)| \ge n(n-f)\,\label{eq:s10}
\end{align}
Denoting by $\xi$ the function that is $1$ if its argument is true, and $0$ otherwise, we may rewrite
\begin{align}
\sum_{i \in [n]}|\In_i(G_1)| & = \sum_{i \in [n]} \sum_{j \in [n]} \xi( j \in \In_i(G_1) ) \notag \\
                             & = \sum_{j \in [n]} \sum_{i \in [n]} \xi( j \in \In_i(G_1) ) \notag \\
                             & = \sum_{j \in [n]} \lvert \{ i \in [n] \mid j \in \In_i(G_1) \}\rvert \notag \\
                             & = \sum_{j \in [n]} \lvert \Out_j(G_1) \rvert \notag
\intertext{and, using \eqref{eq:tocontradict},}
                             & \le nf\,.\notag
\end{align}
Together with \eqref{eq:s10}, we have $n(n-f) \le nf$; a contradiction
  to the assumption that $n > 2f$.
\end{proof}

\section{A Lower Bound for Consensus in Dynamic Networks}

Subsequently, we show that the dynamic radius of a network presents a lower bound on the time
  complexity of a consensus algorithm for this network.

  Let $[n] = \left\{1, \ldots, n \right\}$ be a set of processes that
  operate in lock-step synchronous rounds $r = 1, 2, \ldots$ delimited by times
$t = 0, 1, \ldots$ where, by convention, round $r$ happens between time $r-1$
and time $r$.
Each round consists of a phase of communication, followed by a phase of local
computation.
Like in the previous sections, a communication pattern defines, for each round, which messages reach their
destination.

In the (exact) consensus problem, every node $i \in [n]$ starts with an input value $x_i \in V$
  from an arbitrary domain $V$ and holds a unique write-once variable $y_i$, initialized to
  $y_i = \bot$, where $\bot$ denotes a special symbol s.t.\ $\bot \notin V$.
Since we are concerned with an impossibility result here, we may restrict
  ourselves without loss of generality to the binary consensus problem, i.e.,
  the case where $V = \left\{ 0, 1 \right\}$.
An execution of a deterministic consensus algorithm is a sequence of state
  transitions according to the algorithm and determined by the input assignment
  and the communication pattern.
An algorithm solves consensus if it satisfies in all of its executions:
\begin{itemize}
  \item[](Termination) Eventually for every $i \in [n]$, $y_i \ne \bot$.
  \item[](Validity) If $y_i \ne \bot$ then $y_i = x_j$ for some $j \in [n]$.
  \item[](Agreement) For every $i, j \in [n]$, if $y_i \ne \bot$ and $y_j \ne
    \bot$ then $y_i = y_j$.
\end{itemize}

\begin{thm}\label{thm:consensus}
  If the dynamic radius of the network is $k$, then, in every
  deterministic consensus algorithm, some process has not terminated before
  time $k$.
\end{thm}

\begin{proof}
  Let $\mathcal{G}$ be a communication pattern with dynamic radius $k$,
  which occurs in the network by assumption.
  Suppose, in some deterministic consensus algorithm $A$, all $i \in [n]$ have
  terminated at time $k-1$ in every execution based on $\mathcal{G}$.
  Let $C_0$ be the input assignment where $x_i = 0$ for all $i \in [n]$ and
  $C_1$ be the input assignment where $x_i = 1$ for all $i \in [n]$.
  By validity, when running $A$ under $\mathcal{G}$ and starting from $C_0$, all $i \in
  [n]$ have $y_i = 0$ by time $k-1$ and when starting from $C_1$, they have
  $y_i = 1$.
  Thus, there are input assignments $C, C'$ that differ only in the input
  assignment $x_j$ of a single process $j$ and, for all $i \in [n]$, at time
  $k-1$, $y_i = 0$ when applying $A$ under $\mathcal{G}$ when starting from $C$ and
  $y_i = 1$ when starting from $C'$.
  Since there is no broadcaster in $\mathcal{G}$ before round $k$, there is
  some process $i'$ that did not receive a (transitive) message from $j$ and
  thus $i'$ is in the same state in both executions.
  Therefore, $i'$ decides on the same value in both executions, which
  is a contradiction and concludes the proof.
\end{proof}

\section{Conclusion}

In this paper, we found that nonsplit networks are a convenient abstraction
that arises naturally when considering information dissemination in a variety
of dynamic network settings.
Since classic information dissemination problems are trivially impossible in these
nonsplit dynamic networks, it made sense to study the more relaxed dynamic
radius here.
As we showed in Theorem~\ref{thm:consensus}, this is an important
characteristic with respect to the impossibility of exact consensus.
For our main technical contribution, we proved a new upper bound in
Theorem~\ref{thm:nonsplit:loglog}, which shows that the dynamic radius of
nonsplit networks is in $O(\log \log n)$.
This is an exponential improvement of the previously known upper bound of
$O( \log n)$.

In Section~\ref{sec:async}, we showed an upper bound of $2$ asynchronous
rounds for the dynamic radius in the asynchronous message passing model with
crash failures.
Thus, in this important class of nonsplit networks, information dissemination
is considerably faster than what is currently known for the general case.

Combining our Theorem~\ref{thm:nonsplit:loglog} with the result from
\cite{CFN15:ICALP} that established a $O(n)$ simulation of nonsplit networks in
rooted networks, i.e., networks where every communication graph contains a
rooted spanning tree, yields an improvement of the previously known upper bound
for the dynamic radius of rooted dynamic networks from $O(n \log n)$ to
$O(n \log \log n)$:

\begin{thm}
  The dynamic radius of a dynamic networks whose communication graphs are
  rooted is $O(n \log \log n)$.
\end{thm}

While this is another hint at the usefulness of the nonsplit abstraction for
dynamic networks, the tightness of this bound remains an open question.

%
\bibliographystyle{plain}
\bibliography{references}

\begin{thebibliography}{10}

\bibitem{CFN15:ICALP}
Bernadette Charron-Bost, Matthias F\"ugger, and Thomas Nowak.
\newblock Approximate consensus in highly dynamic networks: The role of
  averaging algorithms.
\newblock In Magn\'us~M. Halld\`orsson, Kazuo Iwama, Naoki Kobayashi, and
  Bettina Speckmann, editors, {\em Automata, Languages, and Programming},
  volume 9135 of {\em Lecture Notes in Computer Science}, pages 528--539.
  Springer Berlin Heidelberg, 2015.

\bibitem{CBS09}
Bernadette Charron-Bost and Andr\'{e} Schiper.
\newblock The {H}eard-{O}f model: computing in distributed systems with benign
  faults.
\newblock {\em Distributed Computing}, 22(1):49--71, April 2009.

\bibitem{ERS66}
Paul Erd{\"o}s, Alfréd Rényi, and Vera~T. Sós.
\newblock On a problem of graph theory.
\newblock {\em Studia Sci. Math. Hungar.}, 1:215--235, 1966.

\bibitem{FL94}
Pierre Fraigniaud and Emmanuel Lazard.
\newblock Methods and problems of communication in usual networks.
\newblock {\em Discrete Applied Mathematics}, 53(1):79 -- 133, 1994.

\bibitem{GKFS10:perf}
Alois Goiser, Samar Khattab, Gerhard Fassl, and Ulrich Schmid.
\newblock A new robust interference reduction scheme for low complexity
  direct-sequence spread-spectrum receivers: Performance.
\newblock In {\em Proceedings 3rd International IEEE Conference on
  Communication Theory, Reliability, and Quality of Service (CTRQ'10)}, pages
  15--21, Athens, Greece, June 2010.

\bibitem{HLPR13}
H.~A. Harutyunyan, A.~L. Liestman, J.~Peters, and D.~Richards.
\newblock Broadcasting and gossiping).
\newblock In J.~Gross and P.~Zhang, editors, {\em The Handbook of Graph
  Theory}, pages 1477--1494. Chapman and Hall/CRC, 2013.

\bibitem{HHL88}
Sandra~M. Hedetniemi, Stephen~T. Hedetniemi, and Arthur~L. Liestman.
\newblock A survey of gossiping and broadcasting in communication networks.
\newblock {\em Networks}, 18(4):319--349, 1988.

\bibitem{HKMP96}
Juraj Hromkovi{\v{c}}, Ralf Klasing, Burkhard Monien, and Regine Peine.
\newblock Dissemination of information in interconnection networks
  (broadcasting {\&} gossiping).
\newblock In Ding-Zhu Du and D.~Frank Hsu, editors, {\em Combinatorial Network
  Theory}, pages 125--212. Springer US, Boston, MA, 1996.

\bibitem{KLO10:STOC}
Fabian Kuhn, Nancy~A. Lynch, and Rotem Oshman.
\newblock Distributed computation in dynamic networks.
\newblock In {\em STOC}, pages 513--522, 2010.

\bibitem{Lyn96}
Nancy~A Lynch.
\newblock {\em Distributed algorithms}.
\newblock Elsevier, 1996.

\bibitem{SW89}
Nicola Santoro and Peter Widmayer.
\newblock Time is not a healer.
\newblock In B.~Monien and R.~Cori, editors, {\em 6th Symposium on Theoretical
  Aspects of Computer Science}, volume 349 of {\em LNCS}, pages 304--313.
  Springer, Heidelberg, 1989.

\bibitem{SBB12}
Udo Schilcher, Christian Bettstetter, and G\"unther Brandner.
\newblock Temporal correlation of interference in wireless networks with
  rayleigh block fading.
\newblock {\em IEEE Transactions on Mobile Computing}, 11(12):2109--2120, 2012.

\bibitem{SZS17}
M.~{Schwarz}, M.~{Zeiner}, and U.~{Schmid}.
\newblock {Linear-Time Data Dissemination in Dynamic Networks}.
\newblock ArXiv e-prints: https://arxiv.org/abs/1701.06800, January 2017.

\bibitem{WJC00}
C.~Ware, J.~Judge, J.~Chicharo, and E.~Dutkiewicz.
\newblock Unfairness and capture behaviour in 802.11 adhoc networks.
\newblock In {\em 2000 IEEE International Conference on Communications. ICC
  2000. Global Convergence Through Communications.}, 2000.

\bibitem{WSS16}
Kyrill {Winkler}, Manfred {Schwarz}, and Ulrich {Schmid}.
\newblock {Consensus in Rooted Dynamic Networks with Short-Lived Stability}.
\newblock {\em arXiv e-prints}, page arXiv:1602.05852, February 2016.

\end{thebibliography}

\end{document}